\documentclass[journal,twoside]{IEEEtran}

\IEEEoverridecommandlockouts                              % This command is only
\usepackage{subfig}
\usepackage{booktabs}
\usepackage{booktabs}
\usepackage{tikz}
\usetikzlibrary{matrix}
\usetikzlibrary{arrows,automata}
\usepackage{amsmath,amssymb,amsthm,url}
\usepackage{graphics}
\usepackage{subfig}
\usepackage{float}
\newtheorem{theorem}{Theorem}
\newtheorem{lemma}{Lemma}

\newtheorem{assumption}{Assumption}

\newtheorem{definition}{Definition}

\usepackage{enumerate}
\usepackage{verbatim}
\usepackage{pgf}
\newcommand{\be}{\begin{equation}}
\newcommand{\ee}{\end{equation}}

\newcommand{\argmax}{\mathop{\rm argmax}}

\usepackage{relsize}
\usepackage{exscale}
\newcommand{\mcup}{\mathsmaller{\bigcup}}
\newcommand{\mcap}{\mathsmaller{\bigcap}}

\def\Pr{\mathbb{P}}

\def\xdn{\mathbf{v}_n}
\def\vn{\mathbf{v}_n}
\def\e{\mathbf{e}}
\usepackage{setspace}
\usepackage{color}
\definecolor{light-gray}{gray}{0.95}

\long\def\old#1{}
\def\op{\overline{p}}
\def\oq{\overline{q}}
\def\u{\mathbf{u}}

\title{On Learning with Finite Memory}

\author{Kimon~Drakopoulos,~\IEEEmembership{Student Member,~IEEE,}   Asuman Ozdaglar,~\IEEEmembership{Member,~IEEE,}  John N. Tsitsiklis,~\IEEEmembership{Fellow,~IEEE}% 
         \thanks{All authors are with the Laboratory of Information and Decision Systems,
                          Massachusetts Institute of Technology,
                          77 Massachusetts Avenue, Room 32-D608,
                          Cambridge, MA 02139. Emails: \{\tt\small kimondr,asuman,jnt\}@mit.edu}
                          \thanks{Research partially supported by Jacobs Presidential Fellowship, the NSF under grant CMMI-0856063, ARO grant W911NF-09-1-0556, AFOSR MURI FA9550-09-1-0538.}}

\begin{document}

\maketitle

\begin{abstract}
We consider an infinite collection of agents {who make} decisions, sequentially, about an unknown underlying binary state of the world.  Each  agent, prior to  making a decision, receives an independent private signal whose distribution depends on the state of the world. Moreover, each agent also observes the decisions of its last $K$ immediate predecessors.  We study conditions  under which the agent decisions converge to the correct value of the underlying state.

We focus on the case where the private signals have bounded information content and investigate whether \textit{learning} is possible, that is, whether there exist decision rules for the different agents  that result in the convergence of their sequence of individual decisions  to the correct state of the world.  We first consider learning in the almost sure sense and show that it is impossible, for any value of $K$.  We then explore the possibility of convergence in probability of the decisions to the correct state. Here, a distinction  arises: if $K=1$, learning in probability is impossible under any decision rule, while for $K\geq 2$, we design a decision rule that achieves it.

We finally consider {a new model, involving} forward looking strategic agents, each of which  maximizes the discounted sum (over all agents) of the probabilities of a correct decision. (The {case, studied in previous literature,} of myopic agents who maximize the probability of their own decision being correct is an extreme special case.) We show that for any value of $K$, for any equilibrium of the associated Bayesian game, and under the assumption that each private signal has bounded information content, learning in probability fails to obtain. 
 \end{abstract}

\section{Introduction} \label{intro}

In this paper, we study variations and extensions of a model introduced and studied in Cover's seminal work \cite{Cov69}. We consider a Bayesian binary hypothesis testing problem over an ``extended tandem'' network architecture whereby each agent $n$ makes a binary decision $x_n$, based on an independent private signal $s_n$ (with a different distribution under each hypothesis) and on the decisions $x_{n-1},\ldots,x_{n-K}$ of its $K$ immediate predecessors, where $K$ is a positive integer constant. We are interested in the question of whether learning is achieved, that is, whether the sequence $\{x_n\}$ correctly identifies the true hypothesis (the ``state of the world,'' to be denoted by $\theta$), almost surely or in probability, as $n\to\infty$. For $K=1$, this coincides with the model introduced by Cover \cite{Cov69} under a somewhat different interpretation, in terms of a single memory-limited agent who acts repeatedly but can only remember its last decision.

At a broader, more abstract level, our work is meant to shed light on the question whether distributed information held by a large number of agents can be successfully aggregated in a decentralized and bandwidth-limited manner. Consider a situation where each of a large number of agents has a noisy signal about an unknown underlying state of the world $\theta$. This state of the world may represent an unknown parameter monitored by decentralized sensors, the quality of a product, the applicability of a therapy, etc. If the individual signals are independent and the number of agents is large, collecting these signals at a central processing unit would be sufficient for inferring (``learning'') the underlying  state $\theta$. However, because of communication or memory constraints, such centralized processing may be impossible or impractical. 
It then becomes of interest to inquire whether $\theta$ can be learned  under 
a decentralized mechanism where each agent communicates a finite-valued summary of its information (e.g., a purchase or voting decision, a comment on the success or failure of a therapy, etc.) to a subset of the other agents, who then refine their own information about the unknown state. 

Whether learning will be achieved under the model that we study depends on various factors, such as the ones discussed next {\color{blue}: } 
\begin{itemize}
\item[(a)] As demonstrated in \cite{Cov69}, the situation is qualitatively different depending on certain assumptions on the information content of individual signals. We will focus exclusively on the case where each signal has bounded information content, in the sense that the likelihood ratio associated with a signal is bounded away from zero and infinity --- the so called Bounded Likelihood Ratio (BLR) assumption. The reason for our focus is that in the opposite case (of unbounded likelihood ratios), the learning problem is much easier; indeed, \cite{Cov69} shows that almost sure learning is possible, even if $K=1$.
\item[(b)] An aspect that has been little explored in the prior literature is the distinction between different learning modes, learning almost surely or in probability. We will see that the results can be different for these two modes.
\item[(c)] The results of \cite{Cov69} suggest that there may be a qualitative difference depending on the value of $K$. Our work will shed light on this dependence.
\item[(d)]  
Whether learning will be achieved or not, depends on the way that agents make their decisions $x_n$. 
In an engineering setting, one can assume that the agents' decision rules are chosen (through an offline centralized process) by a system designer. 
 In contrast, in game-theoretic models, each agent is assumed to be a Bayesian maximizer of an individual  objective, based on the   available  information. 
 Our work will shed light on this dichotomy by considering a special class of individual objectives that incorporate a certain degree of altruism. 
 \end{itemize}

\subsection{Summary of the paper and its contributions}
We provide here a summary of our main results, together with comments on their relation to prior works. 
In what follows, we use the term {\it decision rule} to refer to the mapping from an agent's information to its decision and the term {\it decision profile} to refer to the collection of the agents' decision rules. 
Unless there is a statement to the contrary, all results mentioned below are derived under the BLR assumption.
  \begin{enumerate}
  \item [(a)] {\it Almost sure learning is impossible (Theorem \ref{thm:almsure}).} For any $K\geq 1$, we prove that there exists no 
decision profile that guarantees almost sure convergence of the sequence $\{x_n\}$ of decisions to the state of the world $\theta$. This provides an interesting contrast with the case where the BLR assumption does not hold; in the latter case, almost sure learning is actually possible
\cite{Cov69}.
\item[(b)] {\it Learning in probability is impossible if $K=1$ (Theorem \ref{thm:nolearningwithone}).}  
This strengthens a result of Koplowitz \cite{Kop75} who showed the impossibility of learning in probability for the case where $K=1$ and the private signals $s_n$ are i.i.d.\ Bernoulli random variables.
\item[(c)] {\it Learning in probability is possible if $K\geq 2$ (Theorem \ref{thrm:algorproof}).} 
For the case where $K\geq 2$, we provide a fairly elaborate decision profile that yields learning in probability. 
This result (as well as the decision profile that we construct) is inspired by the positive results in \cite{Cov69} and
\cite{Kop75}, according to which, learning in probability (in a slightly different sense from ours) is possible if each agent can send 4-valued or 3-valued messages, respectively, to its successor. 
In more detail, our construction (when $K= 2$) exploits the similarity between the case of a 4-valued message from the immediate predecessor (as in \cite{Cov69}) and the case of binary messages from the last two predecessors: indeed, the decision rules of two predecessors can be designed  so that their two binary messages convey (in some sense)  information comparable to that in a 4-valued message by a single predecessor. Still, our argument is somewhat more complicated than the ones in \cite{Cov69} and \cite{Kop75}, because in our case, the actions of the two predecessors cannot be treated as arbitrary codewords: they must obey the additional requirement that they equal the correct state of the world with high probability. 

\item[(d)] {\it No learning by forward looking,  altruistic agents
(Theorem \ref{thrm:frwrdlooking}).}
{As already discussed,} when $K\geq 2$, learning is possible, using a suitably designed decision profile. On the other hand, 
if each agent acts myopically (i.e., maximizes the probability that its own decision is correct), 
it is known that learning will not take place (\cite{Cov69,Ban92,Acem10}). To further understand the impact of  selfish behavior, we consider a variation where 
each agent is forward looking, in an altruistic manner: rather than being myopic, each agent takes into account the impact of its  decisions on the error probabilities of future agents. This case can be thought of as an intermediate one, where each agent makes a decision that optimizes its own utility function (similar to the myopic case), but the utility function incentivizes the agent to act in a way that corresponds to good systemwide performance (similar to the case of centralized design). In this formulation, the optimal decision rule of each agent depends on the decision rules of all other agents (both predecessors and successors), which leads to a game-theoretic formulation and a study of the associated equilibria. Our main result shows that 
under any (suitably defined) equilibrium, learning in probability fails to obtain.
In this sense, the forward looking, altruistic setting falls closer to the myopic rather than the engineering design version of the problem. Another interpretation of the result is  that  the carefully designed decision {profile} that can achieve learning will not emerge through the incentives provided by the altruistic model; this is not surprising because the designed decision {profile is} quite complicated. 
\end{enumerate}

\subsection{Outline of the paper}
The rest of the paper is organized as follows. In Section \ref{se:lit}, we review some of the related literature. In Section \ref{model}, we provide a description of our model, notation, and terminology. In Section \ref{chap:almost}, we show that almost sure learning is impossible.
In Section \ref{sec:noinprob} (respectively, Section \ref{sec:learnalgo}) we show that learning in probability is impossible when $K=1$ (respectively, possible when $K\geq 2$). In Section \ref{se:fwd}, we describe the model of forward looking agents and prove the impossibility of learning. We conclude with some brief comments in Section \ref{se:conc}.

   \section{Related literature}\label{se:lit}
    
The literature on information aggregation in decentralized systems is vast; we will restrict ourselves to the discussion of models that involve a Bayesian formulation and are somewhat related to our work. The literature consists of two main  branches, in statistics/engineering  and  in economics. 
    
    \subsection{Statistics/engineering literature}
    A basic version of the model that we consider was studied in the two seminal papers  \cite{Cov69} and \cite{Kop75}, and which have already been discussed in the Introduction. 
     The same model was also studied in \cite{HeCo70}, which gave a characterization of the minimum probability of error, when all agents decide according to the {\it same}  decision rule.
 The case of myopic agents and     
$K=1$ was briefly discussed in \cite{Cov69} who argued that learning (in probability) fails to obtain.  A proof of this negative result was also given in  \cite{AtPa89}, together with the additional result that myopic decision rules will lead to learning if the BLR assumption is relaxed. Finally, \cite{Dia} studies myopic decisions based on private signals and observation of ternary messages from a predecessor in  a  tandem configuration.
        
   Another class  of decentralized information fusion problems was introduced  in \cite{TeSa81}.  In that work, there are again two hypotheses on the state of the world and each one of a set of agents receives a noisy signal regarding the true state. Each agent summarizes its information in a finitely-valued  message which it sends to a fusion center.  The fusion center solves a classical hypothesis testing problem (based on the messages it has received) and decides on one of the two hypotheses. The problem  is the design of decision rules for each agent so as to minimize the probability of error at the fusion center.   A more general network structure, in which each agent  observes  messages from a specific set of agents  before making a decision was introduced in \cite{Ek82} and \cite{EkTe82}, under the assumption that the topology that describes the message flow is a directed tree.  In all of this literature (and under the assumption that the private signals are conditionally independent, given the true hypothesis) each agent's decision rule should be a likelihood ratio test, parameterized by a scalar threshold. However, in general, the problem of optimizing the agent thresholds is a difficult nonconvex optimization problem --- see  \cite{Jnt93} for a survey.
   
  In the line of work initiated in \cite{TeSa81}, the focus is often on tree architectures with large branching factors, so that the probability of error decreases exponentially in the number of sensors. In contrast, for tandem architectures, as in \cite{Cov69,Kop75, AtPa89, Dia}, and for the related ones considered in this paper, learning often fails to hold or takes place at a slow, subexponential rate \cite{TTW08}. The focus of our paper is on this latter class of architectures and the conditions under which learning takes place.

    \subsection{Economics literature}
    
     A number of papers, starting with \cite{Ban92} and \cite{BHW92},  study learning in a setting where each agent, prior to making a decision,  observes the {history of decisions by all} of its predecessors. Each agent is a Bayesian maximizer of the probability that its decision is correct. The main finding is the emergence of ``herds'' or `` informational cascades,'' where agents copy possibly incorrect decisions of their predecessors and ignore their own information, a phenomenon consistent with that discussed by Cover \cite{Cov69} for the tandem model with $K=1$.
The most complete analysis of this framework (i.e., with complete sharing of past decisions) is provided in \cite{SmSo00}, which also draws a distinction between the cases where the BLR assumption holds or fails to hold, and establishes results of the same flavor as those in \cite{AtPa89}. 
   
   A broader class of observation structures is studied in  \cite{SmSo98} and \cite{BaFu04},  with each agent   observing an unordered sample of decisions drawn from the past, namely, the number of sampled predecessors who have taken each of the two actions. The most comprehensive analysis of this setting, where agents are Bayesian but do not observe the full history of past decisions, is provided  in \cite{Acem10}.  This paper considers agents who observe the  decisions of a stochastically generated set of predecessors and provides conditions on the private signals and the network structure under which asymptotic learning  (in probability) to the true state of the world is achieved. 
   
To the best of our knowledge, the first paper that studies forward looking agents is  \cite{SmSo06}: each agent  minimizes the discounted sum of error probabilities of all subsequent agents, including their own. This reference considers the case where the full  past history is observed and shows 
that herding on an incorrect decision is possible, with positive probability. (On the other hand, learning is possible if the BLR assumption is relaxed.)
Finally, \cite{AlKa11} considers a similar model and explicitly characterizes a simple and tractable equilibrium that generates a herd, showing again that even with payoff interdependence and forward looking incentives, {payoff-maximizing agents who observe past decisions can} fail to properly aggregate the available information.

\section{The Model and Preliminaries} \label{model}

In this section we present the observation model (illustrated in Figure \ref{f:model}) and  introduce our basic terminology and notation. 

\subsection{The observation model}
\begin{figure}
\centering
\includegraphics[width=8.5cm]{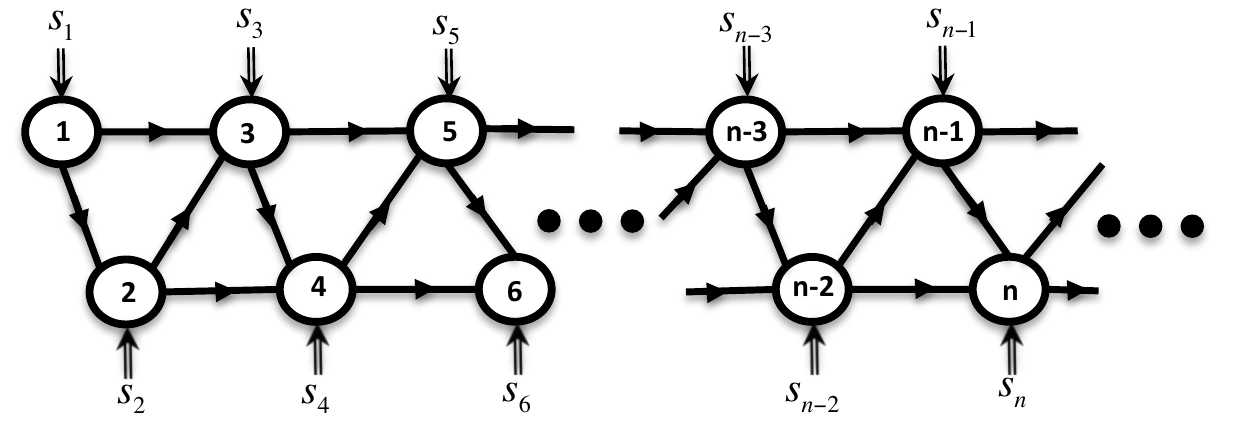}
\caption{\label{f:model}The observation model. If the unknown state of the world is $\theta=j$, $j\in\{0,1\}$, the agents receive independent private signals $s_n$ drawn from a distribution $\mathbb{F}_j$,  and  also observe the decisions of the $K$ immediate predecessors. In this figure, $K=2$.  If agent $n$ observes the decision of agent $k$, we draw an arrow pointing from $k$ to $n$. }
\end{figure}

 We consider an infinite sequence of agents, indexed  by $n \in \mathbb{N}$, where
 $\mathbb{N}$ is the set of natural numbers. There is an underlying \textbf{state of the world} $\theta \in \{0,1\}$, which is modeled as a random variable whose value is unknown by the agents. To simplify notation, we assume that both  of the underlying states are a priori equally likely, that is, $\Pr(\theta=0)=\Pr(\theta=1)={1}/{2}$.  
 
 Each agent $n$ forms posterior beliefs about this state based on  a \textbf{private signal}   that takes values in a set $S$,  and also by observing the decisions of its $K$ immediate predecessors. We  denote by $s_n$ the random variable representing agent $n$'s private signal, while we use $s$ to denote  specific values  in  $S$. 
Conditional on the state of the world $\theta$ being equal to zero (respectively, one), the private  signals are independent random variables distributed  according to a probability measure $\mathbb{F}_0$ (respectively, $\mathbb{F}_1 $) on the set $S$. Throughout the paper, the following two assumptions  will always remain in effect. First, $\mathbb{F}_0$ and $\mathbb{F}_1$ are absolutely continuous with respect to each other, implying that no signal value can be fully revealing about the correct state. Second,  $\mathbb{F}_0$ and $\mathbb{F}_1$ are not identical, so that the private signals can be  informative. 

Each agent $n$ is to make a
\textbf{decision}, denoted by $x_n $, which takes  values in $\{0,1\}$. The information available to agent $n$  consists of its private signal $s_n$
and the random vector 
\[
\mathbf{v}_{n}= (x_{n-K},\ldots,x_{n-1}).
\]
of  decisions of its $K$ immediate predecessors. (For notational convenience an agent $i$ with index $i\leq 0$ is identified with agent 1.)
The decision $x_n$ is made according to a 
\textbf{decision rule}  $d_n:\{0,1\}^K \times S \rightarrow \{0,1\}$: 
$$x_n=d_n(\xdn, s_n).$$ 
  A \textbf{decision  profile} is a sequence $d=\{d_n\}_{n \in \mathbb{N}}$ of decision rules.  Given a decision profile $d$, the sequence $\textbf{x}=\{x_n\}_{n \in \mathbb{N}}$  of agent decisions is a well defined stochastic process, described by  a probability measure to be denoted by $\mathbb{P}_d$, 
  or simply by $\mathbb{P}$ if $d$ has been fixed. 
For notational convenience, we also use  $\Pr^j(\cdot)$  to denote the conditional measure under the state of the world $j$, that is $$\Pr^j(\,\cdot\,)=\Pr(\,\cdot \mid \theta=j).$$
  
It is also useful to consider \textbf{randomized} decision rules, 
whereby the decision $x_n$ is determined according to $x_n=d_n(z_n,\xdn,s_n)$, where $z_n$ is an exogenous random variable which is independent for different $n$ and also independent of $\theta$ and $(\xdn,s_n)$. 
(The construction In Section \ref{sec:learnalgo} will involve a randomized decision rule.)

\subsection{An assumption and the definition of learning}  \label{sec:bu}
As mentioned in the Introduction, we focus on the case where every possible private signal value has bounded information content. The assumption that follows will remain in effect throughout the paper and will not be stated explicitly in our results. 
\begin{assumption}\label{def:BLR} {\rm \textbf{(Bounded Likelihood Ratios --- BLR).}} There exist some $m>0$ and  $M < \infty$,  such that the Radon-Nikodym derivative ${d\mathbb{F}_0}/{d\mathbb{F}_1}$ satisfies
\[
m < \frac{d\mathbb{F}_0}{d\mathbb{F}_1}(s) < M,
\]
for almost all $s \in S$ under the measure $(\mathbb{F}_0 +\mathbb{F}_1)/2$ 
%\footnote{ We say that the signal structure has {\textbf{Unbounded Likelihood Ratios}} if the essential infimum  of $d\mathbb{F}_0/d\mathbb{F}_1(s)$  is $0,$ while  the essential supremum of $d\mathbb{F}_0/d\mathbb{F}_1(s)$ is infinity,  under the  measure  $(\mathbb{F}_0 +\mathbb{F}_1)/2$.}.
\end{assumption}

We study two different types of learning. As will be seen in the sequel, the results for these two types are, in general, different. 
\begin{definition}
We say that a decision profile  $d$ achieves  \textbf{almost sure learning} if
\[
\lim_{n\rightarrow \infty}x_n = \theta, \qquad \Pr_d\text{-almost surely},
\]
and that it  achieves  \textbf{learning in probability}  if 
\[
\lim_{n\rightarrow \infty} \mathbb{P}_{d}(x_n=\theta)=1. 
\]
\end{definition}

\section{{Impossibility of} almost sure learning}\label{chap:almost}
In this section, we show that almost sure learning is impossible, for any value of $K$.

\begin{theorem}\label{thm:almsure}
For any given number  $K$ of observed {immediate predecessors},  there  exists no decision profile  that achieves almost sure learning.\end{theorem} 

The rest of this section is devoted to the proof of Theorem~\ref{thm:almsure}. We note that 
the proof  does not use anywhere the fact that each agents only observes the last $K$ immediate predecessors. The exact same proof establishes the impossibility of almost sure learning even for a more general model where each agent $n$ observes the decisions of an arbitrary subset of its predecessors. 
Furthermore, while the proof is given for the case of deterministic decision rules, the reader can verify that {it also} applies to the case where randomized decision rules are allowed.

The following lemma  is  a simple consequence of the BLR assumption and its proof is omitted.

\begin{lemma} \label{lem:balancing}
For any $\mathbf{u} \in \{0,1\}^K$ and any $j \in \{0,1\}$, we have
\begin{align} 
m\cdot\mathbb{P}^1(x_n=j \mid \,\xdn=\mathbf{u}) \nonumber &<{\mathbb{P}^0(x_n=j \mid \xdn=\mathbf{u} )}\nonumber\\ 
&<M\cdot{\mathbb{P}^1(x_n=j \mid \xdn=\mathbf{u})},
\end{align}
where  $m$ and $M$ are as in Definition \ref{def:BLR}.
\end{lemma}

Lemma \ref{lem:balancing} states that (under the BLR assumption) if under one state of the world some agent $n$, after observing $\mathbf{u}$, decides $0$  with positive probability, then the same must be true with proportional probability under the other state of the world. This  proportional dependence of decision probabilities for the two  possible underlying states is central to the proof of Theorem~$1$.

Before proceeding with the main part of the proof, we need two more lemmata. Consider a probability space $(\Omega,\mathcal{F},\Pr)$ and a sequence of events $\{E_k\}$, $k=1,2, \ldots $. The upper limiting set of the sequence, $\limsup_{k \rightarrow \infty} E_k$, is defined by 
\[
\limsup_{k \rightarrow \infty} E_k = \mcap_{n=1}^{\infty} \mcup_{k=n}^\infty E_k.
\]
{(This is the event that infinitely many of the $E_k$ occur.)} 
We will use a variation of the Borel-Cantelli lemma (Corollary 6.1 in \cite{Durett}) that does not require independence of events.

\begin{lemma}\label{lem:condbc}
If 
\[
\sum_{k=1}^{\infty}\Pr(E_k\mid E_1' \ldots E_{k-1}')=\infty,  
\]
then,
\[
\Pr\left(\limsup_{k \rightarrow \infty} E_k\right )=1,
\]
where $E'_k$ denotes the complement of $E_k$.
\end{lemma}

\noindent Finally, we will use  the following algebraic fact.
\begin{lemma}\label{l:ineq}
Consider a sequence $\{q_n\}_{n \in \mathbb{N}}$ of  real numbers, with $q_n \in [0,1]$, for all $n \in \mathbb{N}$.  Then,
\[
1-\sum_{n \in V} q_n \leq \prod_{n \in V} (1-q_n)\leq e^{-\sum_{n \in V}q_n},
\]
for any $V\subseteq \mathbb{N}$. 
\end{lemma}
\begin{proof} The second inequality is standard. For the first one, interpret the numbers  $\{q_n\}_{n \in \mathbb{N}}$ as probabilities of  independent events $\{E_n\}_{n \in \mathbb{N}}$. Then, clearly,
\[
\Pr(\mcup_{n \in V} E_n)+\Pr(\mcap_{n \in V} E'_n)=1.
\]
Observe that 
\[
\Pr(\mcap_{n \in V} E'_n)=\prod_{n \in V} (1-q_n),
\]
and by the union bound,
\[
\Pr(\mcup_{n \in V} E_n)\leq \sum_{n \in V} q_n.
\]
 Combining the above yields the desired result.
\end{proof}

\noindent We are now ready to prove the main result of this section.

\begin{proof}[Proof of Theorem \ref{thm:almsure}]
Let $U$ denote the set of all binary sequences with a finite number of zeros (equivalently, the set of binary sequences that converge to one). 
Suppose, to derive a contradiction, that we have almost sure learning. Then, 
$\Pr^1(\textbf{x} \in U)=1$. 
The set $U$ is easily seen to be countable, which implies that there exists an infinite binary sequence
$\mathbf{u}=\{u_n\}_{n\in \mathbb{N}}$ such that
$\Pr^1(\mathbf{x}=\mathbf{u})>0$.
In particular,
\[
\Pr^1(x_k={u}_k, \text{  for all  } k< {n})>0,  \ \ \ \text{   for all    }  n \in \mathbb{N}. \]
Since $(x_1, x_2,\ldots, x_n)$ is determined by $(s_1, s_2, \ldots, s_n)$ and since the distributions of $(s_1, s_2, \ldots, s_n)$ under the two hypotheses are absolutely continuous with respect to each other, it follows that 
\begin{equation}
\Pr^0(x_k={u}_k, \text{  for all  } k\leq {n})>0,  \ \ \ \text{   for all    }  n \in \mathbb{N} \label{eq:posit2}.
\end{equation}
We define 
\begin{equation*}
a_n^0=\mathbb{P}^0(x_n \neq u_n \mid x_k=u_k,  \text{   for all    } k < n ),
\end{equation*}
\begin{equation*}
a_n^1=\mathbb{P}^1(x_n \neq u_n \mid x_k=u_k,  \text{   for all    } k < n ).
\end{equation*}
Lemma \ref{lem:balancing} implies that 
\begin{equation}
m{a_n^1}< {{a}_n^0}< M{{a}_n^1}, \label{eq:balancing2}
\end{equation}
because for $j \in \{0,1\}, $ $\mathbb{P}^j(x_n \neq u_n \mid x_k=u_k,  \text{   for all    } k < n )=\mathbb{P}^j(x_n \neq u_n \mid x_k=u_k, \text{  for  } k=n-K,\ldots,n-1)$.

Suppose that 
\[
 \sum_{n=1}^\infty{a}_n^1 = \infty.
 \]
Then, Lemma \ref{lem:condbc}, with the identification $E_k=\{x_k\neq u_k\}$,  implies 
that the event $\{x_k\neq u_k,\ \text{ for some }k\}$ has probability 1, under $\Pr^1$.
Therefore, $\Pr^1(\mathbf{x}=\mathbf{u})=0$, which contradicts the definition of $\mathbf{u}$.

Suppose now that
$
\sum_{n=1}^\infty{a}_n^1< \infty.
$
Then,
\[
\sum_{n=1}^\infty {a}_n^0 <M\cdot \sum_{n=1}^\infty {a}_n^1< \infty,
\]
and
\begin{align*}
\lim_{N\rightarrow \infty} \sum_{n=N}^\infty & \mathbb{P}^0(x_n \neq u_n \mid x_k=u_k,  \, \text{   for all    }   k < n )\\
&=\lim_{N\rightarrow \infty} \sum_{n=N}^\infty {a}_n ^0 =0.
 \label{eq:zerosum}
\end{align*}
Choose some $\hat{N}$ such that 
\[
\sum_{n=\hat{N}}^\infty \mathbb{P}^0(x_n \neq u_n \mid x_k=u_k,  \text{   for all    } k < n )<\frac{1}{2}.
\]
 Then, 
\begin{align*}
\Pr^0&(\mathbf{x}= \mathbf{u})=\Pr^0(x_k=u_k, \text{  for all   } k< \hat{N}) \\
& \cdot \prod_{n=\hat{N}}^{\infty}(1-\mathbb{P}^0(x_n \neq u_n \mid x_k=u_k,  \text{   for all    } k < n )).
\end{align*}
The first term on the right-hand side is positive  by (\ref{eq:posit2}), while 
\begin{align*}
\prod_{n=\hat{N}}^{\infty}&(1-\mathbb{P}^0(x_n \neq u_n \mid x_k=u_k, \text{   for all    } k < n ))\\
 &\geq1-\sum_{n=\hat{N}}^{\infty}\mathbb{P}^0(x_n \neq u_n \mid x_k=u_k,  \text{   for all    } k < n ) > \frac{1}{2}.
\end{align*}
Combining the above, we obtain $\Pr^0(\mathbf{x}= \mathbf{u})>0$ and  
\[
\liminf_{n\to\infty} \Pr^0(x_n=1) \geq \Pr^0(\mathbf{x}= \mathbf{u})>0,
\]
which contradicts almost sure learning and completes the proof.
\end{proof}

Given Theorem~\ref{thm:almsure}, in the rest of the paper we concentrate exclusively on the weaker notion of learning in probability, as defined in Section \ref{sec:bu}.

\section{No learning in  probability when $K=1$}\label{sec:noinprob}
In this section, we consider
the case where $K=1$,  so that each agent only observes the decision of its immediate predecessor.
Our main result, stated next, shows that learning in probability is not possible.

\begin{theorem}\label{thm:nolearningwithone}
If $K=1$, there exists no  decision profile  that achieves learning in probability.
\end{theorem}

We fix a decision profile and use a Markov chain to represent the evolution of the decision process under  a particular state of the world. In particular, we consider a two-state Markov chain whose state is the observed decision $x_{n-1}$. A transition from state $i$ to state $j$ for the Markov chain associated with $\theta=l$, where $i,j,l \in \{0,1\},$ corresponds   to agent $n$ taking the decision $j$ given that its immediate predecessor $n-1$  decided $i$, under the state $\theta=l$. The Markov property is satisfied because the decision $x_n$, conditional on the immediate predecessor's decision, is determined by $s_n$ and  hence is {(conditionally)} independent from the history of previous decisions.
Since a decision profile $d$ is fixed, we can again suppress $d$ from our notation and define the transition probabilities of the two chains by
\begin{eqnarray}
a_n^{ij}&=&\mathbb{P}^0(x_n=j \mid x_{n-1} = i )\\
{\bar{a}}_n^{ij}&=&\mathbb{P}^1(x_n=j \mid x_{n-1} = i),
\end{eqnarray}
where $ i,j \in \{0,1\}$.
The two chains are illustrated in
 Fig.~\ref{fig:markovchains}.
Note that in the current context,  and similar to Lemma \ref{def:BLR}, the BLR assumption yields the inequalities
 \begin{equation}\label{cor:blrlemma}
 m\cdot {{\bar{a}}_n^{ij}}<{a_n^{ij}}< M\cdot{{\bar{a}}_n^{ij}}, 
 \end{equation}
where $ i,j \in \{0,1\}$, and $m>0$, $M< \infty$, are as in Definition~\ref{def:BLR}.

\begin{figure}
\centering
\includegraphics[width=4cm]{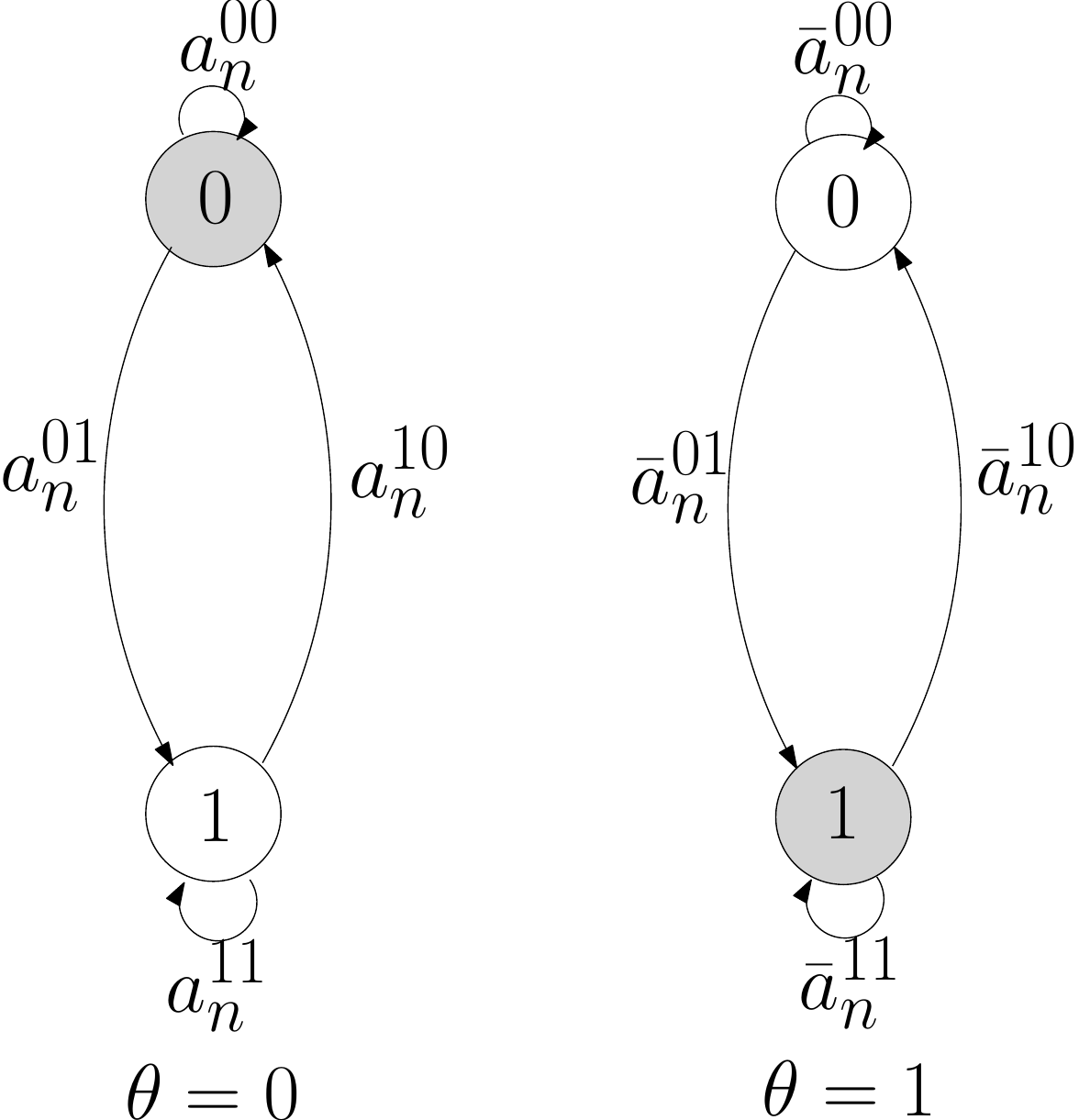}
\caption{The Markov chains that model the decision process for $K=1$. States represent observed decisions. {The} transition probabilities under $\theta=0$ or $\theta=1$ are given by $a_n^{ij}$ and $\bar{a}_n^{ij}$, respectively. If learning in probability is to occur, the probability mass needs to become concentrated on the highlighted state.}\label{fig:markovchains}

\end{figure}

We now establish a further relation between the transition probabilities under the two states of the world.

 \begin{lemma}\label{lem:sum} 
If we have learning in probability, then
  \begin{equation}
  \sum_{n=1}^{\infty} a_n^{01}=\infty,
  \end{equation}
  and 
    \begin{equation}
      \sum_{n=1}^{\infty} a_n^{10}=\infty.
      \end{equation}
  \end{lemma}
  
\begin{proof}
For the sake of contradiction, assume that  $\sum_{n=1}^{\infty} a_n^{01}<\infty$. By Eq.~\ref{cor:blrlemma}, we also have  $\sum_{n=1}^{\infty} \bar{a}_n^{01}<\infty$. Then, the expected number of transitions from state $0$ to state $1$ is finite under either state of the world. In particular the (random) number of such transitions is finite, almost surely. This can only happen if $\{x_n\}_{n=1}^\infty$ converges almost surely. However, almost sure convergence together with learning in probability would imply almost sure learning, which would contradict Theorem~\ref{thm:almsure}. 
The proof of the second statement in the lemma is similar. 
\end{proof}

The next lemma states that if we have  learning in probability, then the transition probabilities between different states should converge to zero. 

\begin{lemma}\label{lem:convzero} 
If we have learning in probability, then
  \begin{equation}
  \lim_{n \rightarrow \infty} a_n^{01} =0.
  \end{equation}
  \end{lemma}
  
 \begin{proof}
   Assume, to arrive at a contradiction that there exists some $\epsilon \in (0,1)$ 
   such that 
   $$a_n^{01}=\Pr^0(x_{n}=1 \mid x_{n-1}=0)>\epsilon,$$
   for infinitely many values of $n$. Since we have learning in probability, 
   we also have $\Pr^0(x_{n-1}=0)>1/2$ when $n$ is large enough. This implies that for infinitely many values of $n$,
$$\Pr^0(x_{n}=1)\geq   \Pr^0(x_{n}=1 \mid x_{n-1}=0) \Pr^0(x_{n-1}=0)
\geq \frac{\epsilon}{2}.$$ But this contradicts learning in probability.
\end{proof}

We are now ready to complete the proof of Theorem~\ref{thm:nolearningwithone}, by arguing as follows.  Since the transition probabilities from state $0$ to state $1$ converge to zero, while their sum is infinite, under either state of the world, we can divide the agents (time) into blocks so that  the corresponding sums of the transition probabilities from state 0  to state 1 over each block are approximately constant. If  during such a block the sum  of the  transition  probabilities from state 1 to state 0 is large, then under the state of the world $\theta=1$, there is high probability of starting the block at state 1, moving to state 0, and staying at state 0 until the end of the block. 
If on the other hand  the sum of the transition probabilities from state 1 to state 0 is small, then under state of the world $\theta=0$, there is high probability of starting the  block at state 0, moving to state 1, and staying at state 1 until the end of the block.
Both cases prevent  convergence in probability to the correct decision.

\begin{proof}[Proof of Theorem \ref{thm:nolearningwithone}]
We assume that we have learning in probability and will derive a
 contradiction.  From Lemma~\ref{lem:convzero}, $\lim_{n \rightarrow \infty} a_n^{01}=0$ and therefore there exists a $\hat{N} \in \mathbb{N}$ such that for all $n>\hat{N},$ 
\begin{equation} \label{eq:one}
a_n^{01} < \frac{m}{6}.
\end{equation}
Moreover, by the learning in probability assumption,  there exists some $\tilde{N} \in \mathbb{N}$ such that for all $n>\tilde{N}$,
\begin{equation} \label{eq:two}\mathbb{P}^0(x_n=0)>\frac{1}{2}, \end{equation}
and
\begin{equation} \label{eq:three} \mathbb{P}^1(x_n=1)>\frac{1}{2}.\end{equation}
 \noindent Let $N=\max\{\hat{N},\tilde{N}\}$ so that Eqs. (\ref{eq:one})-(\ref{eq:three}) all hold for
$n>N$, 

We divide the agents (time) into blocks so that in each block the sum of the transition probabilities from state $0$ to state $1$ can be simultaneously bounded from above and below. We define the last agents of each block recursively, as follows: 
\begin{eqnarray*}
r_1&=&{N},\\
r_k&=&\min\left\{l:\sum_{n=r_{k-1}+1}^l a_n^{01}\geq \frac{m}{2}\right\}.
\end{eqnarray*}
From Lemma \ref{lem:sum},   we have  that $\sum_{n=N}^{\infty} a_n^{01}=\infty$. This fact, together with Eq.~(\ref{eq:one}), guarantees that the sequence $r_k$ is well defined and strictly increasing.

Let $A_k$ be the block that ends with agent $r_{k+1}$, i.e.,  $A_k \triangleq\{r_k+1,\ldots, r_{k+1}\}$. The construction of the sequence $ \{r_k\}_{k\in\mathbb{N}}$ yields 
\[
\sum_{n\in A_k} a_n^{01} \geq \frac{m}{2}.
\]
On the other hand, $r_{k+1}$ is the first agent for which the sum is at least  $m/2$ and since, by (\ref{eq:one}),  $a_{r_{k+1}}<m/6$, we get that 
\[
\sum_{n\in A_k} a_n^{01} \leq \frac{m}{2}+\frac{m}{6}=\frac{2m}{3}.
\]
Thus,
\begin{equation}
\frac{m}{2}\leq\sum_{n \in A_k}a_n^{01}  \leq   \frac{2m}{3},
\end{equation}
and
combining with Eq.~\eqref{cor:blrlemma}, we also have
 \begin{equation}
\frac{m}{2M}\leq\sum_{n \in A_k}{\bar{a}}_n^{01}  \leq   \frac{2}{3}, \label{eq:01}
\end{equation}
for all $k$.

We consider  two cases for the sum of transition probabilities from state $1$ to state $0$ during block $A_k$.
 We first assume that 
\[ \sum_{n \in A_k} a_n^{10}>\frac{1}{2}. \]
 Using Eq.~\eqref{cor:blrlemma},  we obtain 
\begin{equation} \label{eq:otherworld}
\sum_{n \in A_k} \bar{a}_n^{10}>\sum_{n \in A_k} \frac{1}{M}\cdot {a}_n^{10}>\frac{1}{2M}.
\end{equation}
The probability of a transition from state $1$ to state $0$  during the block $A_k$, under $\theta=1$, is
\begin{equation*}
\Pr^1\left(\mcup_{n\in A_k} \{x_n=0\}\mid  x_{r_k}=1\right)=1-\prod_{n\in A_k}(1-{\bar{a}}_n^{10})
\end{equation*}

Using Eq.~(\ref{eq:otherworld}) and Lemma \ref{l:ineq}, the product on the right-hand side can be bounded from above, 
\[
\prod_{n\in A_k}(1-{\bar{a}}_n^{10}) \leq e^{-\sum_{n \in A_k} \bar{a}_n^{10}}\leq e^{-{1}/{(2M)}},
\]
which yields
\[
\Pr^1\left( \mcup_{n\in A_k} \{x_n=0\}\mid x_{r_k}=1\right) \geq 1- e^{-{1}/{(2M)}}.
\]
After a transition to state $0$ occurs, the probability of staying at that state until the end of the block is bounded  below as follows:
\[
\Pr^1\left( x_{r_{k+1}}=0\mid \mcup_{n\in A_k} \{x_n=0\}\right) \geq \prod_{n\in A_k}(1-{\bar{a}}_n^{01}).
\]
The right-hand side can be further bounded using Eq.~(\ref{eq:01}) and Lemma~\ref{l:ineq}, as follows:
\[
\prod_{n\in A_k}(1-{\bar{a}}_n^{01})\geq 1-\sum_{n \in A_k} \bar{a}_n^{01}\geq \frac{1}{3}.
\]
Combining the above and using (\ref{eq:three}), we conclude that 
\begin{align*}
\Pr^1(x_{r_{k+1}}=0)\geq &\, \Pr^1(x_{r_{k+1}}=0\mid  \mcup_{n\in A_k} \{x_n=0\}) \\
&\cdot \Pr^1(\mcup_{n\in A_k} \{x_n=0\}\mid x_{r_k}=1) \Pr^1(x_{r_{k}}=1)\\
\geq& \frac{1}{3}\cdot\left( 1-e^{-{1}/{(2M)}}\right)  \cdot \frac{1}{2}.
\end{align*}

We now consider the second case and assume that 
\[ \sum_{n \in A_k} a_n^{10}\leq\frac{1}{2}. \]
The probability of a transition from state $0$ to state $1$ during the block $A_k$ is 
\[
\Pr^0 \left(\mcup_{n\in A_k} \{x_n=1\}\mid x_{r_k}=0\right)=1-\prod_{n\in A_k}(1-{{a}}_n^{01}).
\]
The product on the right-hand side can be bounded above using Lemma~\ref{l:ineq},
\[
\prod_{n\in A_k}(1-{{a}}_n^{01}) \leq e^{-\sum_{n \in A_k} {a}_n^{01}}\leq e^{-{m}/{(2M)}},
\]
which yields
\[
\Pr^0\left( \mcup_{n\in A_k} \{x_n=1\}\mid x_{r_k}=0\right) \geq 1- e^{-{m}/{2}}.
\]
After a transition to state $1$ occurs, the probability of staying at that state until the end of the block is bounded from below as follows:
\[
\Pr^0\left( x_{r_{k+1}}=1\mid \mcup_{n\in A_k} \{x_n=1\}\right) \geq \prod_{n\in A_k}(1-{{a}}_n^{10}).
\]
The right-hand side can be bounded using Eq.~(\ref{eq:01}) and Lemma~\ref{l:ineq}, as follows:
\[
\prod_{n\in A_k}(1-{{a}}_n^{10})\geq 1-\sum_{n \in A_k} {a}_n^{10}\geq \frac{1}{2}.
\]
Using also Eq.~\eqref{eq:two}, we conclude that 
\begin{align*}
\Pr^0(x_{r_{k+1}}=1) \geq&\, \Pr^0(x_{r_{k+1}}=1\mid  \mcup_{n\in A_k} \{x_n=1\})\\
&\cdot \Pr^0(\mcup_{n\in A_k} \{x_n=1\}\mid x_{r_k}=0)\Pr^0(x_{r_{k}}=0) \\
\geq& \frac{1}{2}\cdot \left( 1-e^{-{m}/{2}}\right) \cdot \frac{1}{2}.
\end{align*}

Combining the two cases we conclude that 
\begin{align}
\liminf_{n\rightarrow \infty}&\,\mathbb{P}_d(x_n\neq \theta) \label{eq:result} \\
& \geq\frac{1}{2} \min\left\{\frac{1}{6} \left( 1-e^{-{1}/{(2M)}}\right) ,\frac{1}{4}\left( 1-e^{-{m}/{2}}\right) \right\} >0 \nonumber
\end{align}
which contradicts learning in probability and concludes the proof. 
\end{proof}

Once more, we note that the proof and the result remain valid for the case where randomized decision rules are allowed.

The coupling between the Markov chains associated with the two states of the world is central to the proof of Theorem~\ref{thm:nolearningwithone}. The importance of the BLR assumption is highlighted by the observation that if either $m=0$ or $M=\infty$, then the lower bound obtained in (\ref{eq:result}) is zero, and the proof fails. 
The next  section shows that a similar argument cannot be made to work when $K\geq 2$. In particular, we  construct a decision profile that achieves learning in probability when agents observe  the last two  immediate predecessors.

\section{Learning in probability when $K\geq 2$} \label{sec:learnalgo}

In this section we show that learning in probability is possible when $K\geq 2$, i.e., when each agent  observes the decisions of two or more of its immediate predecessors.

\subsection{Reduction to the case of binary observations}

We will construct a decision profile that leads to learning in probability, for the special case where the signals $s_n$ are binary (Bernoulli) random variables with a different parameter under each state of the world. This readily leads to a decision profile that learns, for the case of general signals. Indeed, if the $s_n$ are general random variables, each agent can quantize its signal, to obtain a quantized signal $s'_n=h(s_n)$ that takes values in $\{0,1\}$. Then, the agents can apply the decision profile for the binary case. The only requirement is that the distribution of $s'_n$ be different under the two states of the world. This is straightforward to enforce by proper choice of the quantization rule $h$: for example, we may  let $h(s_n)=1$ if and only if $\Pr(\theta=1\mid s_n)> \Pr(\theta=0 \mid s_n)$. It is not hard to verify that with this construction and under our assumption that the distributions $\mathbb{F}_0$ and $\mathbb{F}_1$ are not identical, the distributions of $s_n'$ under the two states of the world will be  different.

We also note that it suffices to construct a decision profile for the case where $K=2$. Indeed, if $K>2$, we can have the agents ignore the actions of all but their two immediate predecessors and employ the decision profile designed for the case where $K=2$.

 \subsection{The decision profile}\label{sec:algorithm}
As just discussed, we assume that the signal $s_n$ is binary. For $i=0,1$, we let $p_i=\Pr^i(s_n=1)$ and
$q_i=1-p_i$. We also use $p$ to denote a random variable that is equal to $p_i$ if and only if  $\theta=i$. Finally, we let $\op=(p_0+p_1)/2$ and $\oq=1-\op=(q_0+q_1)/2$.
We assume, without loss of generality, that $p_0<p_1$, in which case we have $p_0<\op<p_1$
and $q_0>\oq>q_1$.

Let $\{k_m\}_{m\in\mathbb{N}}$ and $\{r_m\}_{m\in\mathbb{N}}$  be two sequences of positive integers that we will define later in this section. 
We divide the agents into segments that consist of  S-blocks, R-blocks, and transient agents,  as follows. We do not assign the first two agents to any segment (and the first segment starts with agent $n=3$).
For segment $m \in \mathbb{N}$: 
 \begin{enumerate}[(i)]
 \item   the first  $2k_m-1$ agents belong to the block $S_m$; 
 \item  the next agent is an SR transient agent;
 \item  the   next $2r_m-1$ agents belong to the block $R_m$;
 \item the next agent is an RS transient agent.
 \end{enumerate}

An agent's information consists of the last two decisions, denoted by $\xdn=(x_{n-2},x_{n-1})$, and
its own signal $s_n$.
The decision profile is constructed so as to enforce that if $n$ is the first agent of either an S or R block, then $\xdn=(0,0)$ or $(1,1)$.

\begin{enumerate}[(i)]

\item Agents 1 and 2 choose 0, irrespective of their private signal.
\item{During  block $S_m$, for $m\geq 1$:}
\begin{enumerate}
\item  If the first agent of the block, denoted by $n$, observes $(1,1)$, it chooses $1$, irrespective of its private signal. If it observes $(0,0)$  and its private signal is $1$, then
\[
x_n=z_n,
\]
where $z_n$ is an independent Bernoulli random variable with parameter $1/m$. 
{If $z_n=1$ we say that a \textbf{searching phase is initiated.}}
(The cases of observing $(1,0)$ or $(1,0)$ will not be allowed to occur.)  

\item  For the remaining agents in the block:
\begin{enumerate} 

\item Agents who observe $(0,1)$ decide $0$ for all private signals.

\item Agents who observe $(1,0)$ decide {$1$ if and only if their private signal is $1$.}

\item Agents who observe $(0,0)$ decide $0$ for all private signals.

\item Agents who observe $(1,1)$ decide $1$ for all private signals.

\end{enumerate}
\end{enumerate}

\item During block $R_m:$ 
\begin{enumerate}
\item   If the first agent of the block, denoted by $n$, observes $(0,0)$, it chooses $0$, irrespective of its private signal. If it observes $(1,1)$  and its private signal is $0,$ then
\[
x_n=1-z_n,
\]
where $z_n$ is a Bernoulli random variable with parameter $1/m$. 
{If $z_n=1$,  we say that a \textbf{searching phase is initiated.}}
(The cases of observing $(1,0)$ or $(0,1)$ will not be allowed to occur.)  

\item  For the remaining agents in the block:
\begin{enumerate} 

\item Agents who observe $(1,0)$ decide $1$ for all private signals.

\item Agents who observe $(0,1)$ decide {$0$ if and only if their private signal is $0$.}

\item Agents who observe $(0,0)$ decide $0$ for all private signals.

\item Agents who observe $(1,1)$ decide $1$ for all private signals.

\end{enumerate}
\end{enumerate}

\item {An SR or RS transient agent $n$ sets $x_n=x_{n-1}$, irrespective of its  private signal.}

\end{enumerate}

 \begin{figure}
  \centering
\subfloat[The decision rule for the first agent of block $S_m$.]{\includegraphics[width=4cm]{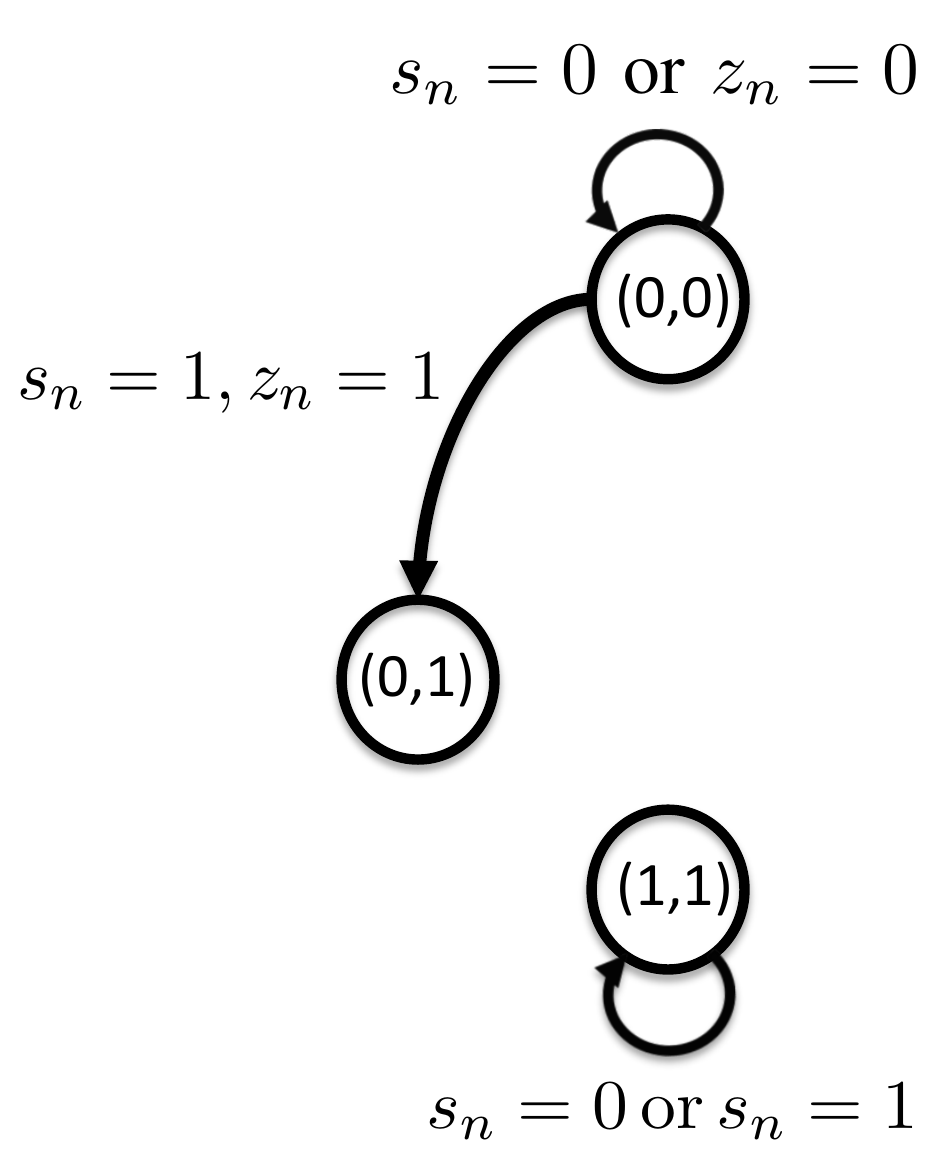}}
\qquad
\subfloat[The decision rule for all agents  of block $S_m$ but the first.]{\includegraphics[width=4cm]{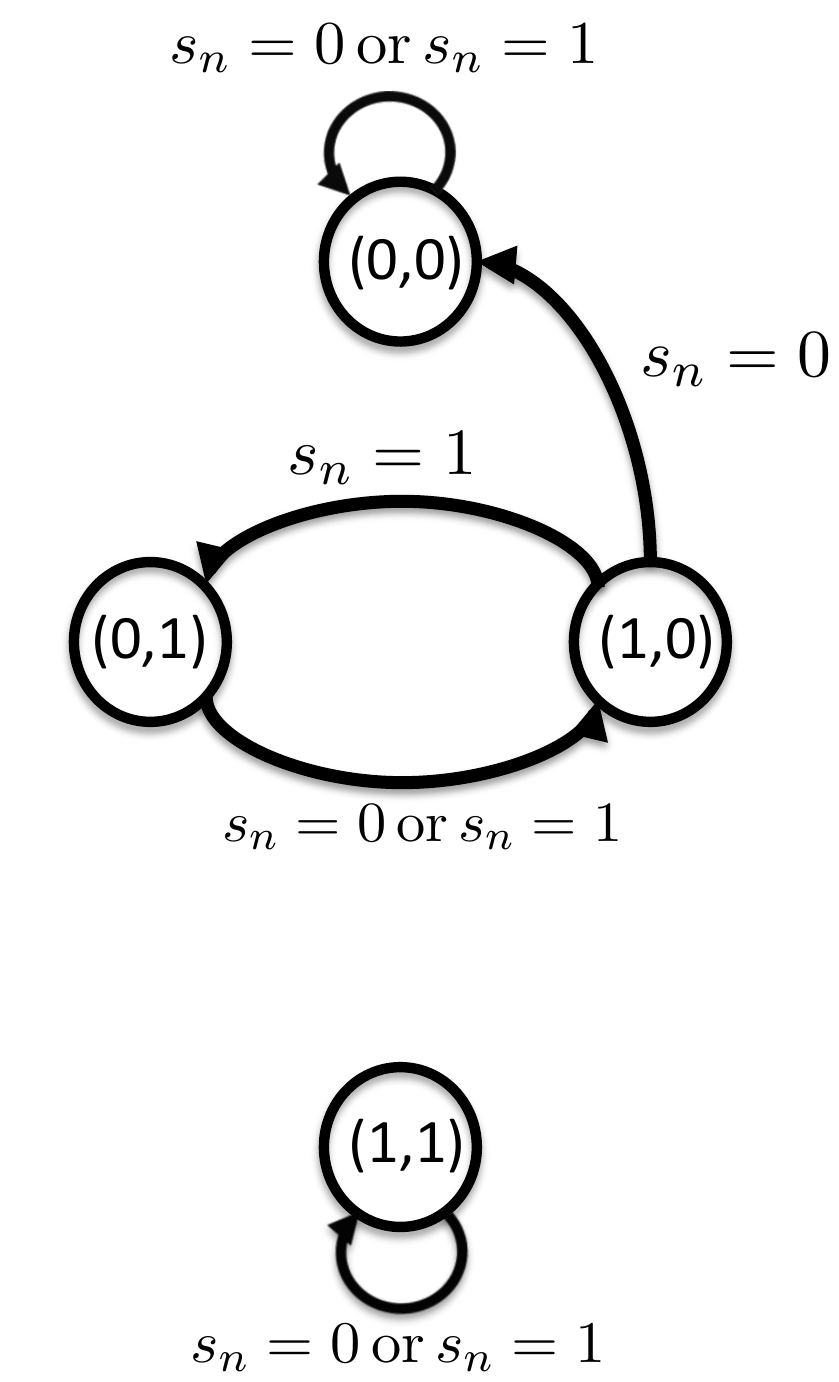}}
\caption{Illustration of the decision profile during block $S_m$. Here, $z_n$ is a Bernoulli random variable, independent from $s_n$ or $\xdn$,  which takes the value $z_n=1$ with a small probability ${1}/{m}$. In this figure, the state represents the decisions of the last two agents and the decision rule dictates the probabilities of transition between states.
 }\label{fig:decrule}
  \end{figure}

We now discuss the 
evolution of the decisions (see also Figure \ref{fig:decrule} for an illustration of the different transitions). We first note that because $\mathbf{v}_3=(x_1,x_2)=(0,0)$ and because of the rules for transient agents,  our requirement that $\xdn$ be either $(0,0)$ or $(1,1)$ when $n$ lies at the beginning of a block, is automatically satisfied. Next, we discuss the possible evolution of $\xdn$ in the course of a block $S_m$. (The case of a block $R_m$ is entirely symmetrical.) Let $n$ be the first agent of the block, and note that the last agent of the block is $n+2k_m-2$.
\begin{enumerate}

\item
If $\mathbf{v}_n=(1,1)$, then $\mathbf{v}_i=(1,1)$ for all agents $i$ in the block, as well as for the subsequent SR transient agent, which is agent {$n+2k_m-1$}. The latter agent also decides $1$, so that the first agent of the next block, $R_m$, observes $\mathbf{v}_{n+2k_m}=(1,1)$.

\item
If $\mathbf{v}_n=(0,0)$ and $x_n=0$, then $\mathbf{v}_i=(0,0)$ for all agents $i$ in the block, as well as for the subsequent SR transient agent, which is agent $n+2k_m-1$. The latter agent also decides $0$, so that the first agent of the next block, $R_m$, observes $\mathbf{v}_{n+2k_m}=(0,0)$.

\item The interesting case occurs when $\mathbf{v}_n=(0,0)$, $s_n=1$, and $z_n=1$, so that a search phase is initiated and $x_n=1$, $\mathbf{v}_{n+1}=(0,1)$, $x_{n+1}=0$, $\mathbf{v}_{n+2}=(1,0)$.
Here there are two possibilities:

\begin{itemize}

\item[a)]
Suppose that for every $i>n$ in the block $S_m$, for which $i-n$ is even {(and with $i$ not the last agent in the block)}, we have $s_i=1$. Then, for $i-n$ even, we will have $\mathbf{v}_{i}=(1,0)$, $x_i=1$, $\mathbf{v}_{i+1}=(0,1)$, $x_{i+1}=0$,
$\mathbf{v}_{i+2}=(1,0)$, etc. When $i$ is the last agent of the block, then $i=n+2k_m-2$, so that
$i-n$ is even, $\mathbf{v}_{i}=(1,0)$, and $x_i=1$. The subsequent SR transient agent, agent $n+2k_m-1$, sets $x_{n+2k_m-1}=1$, so that the first agent of the next block, $R_i$, observes $\mathbf{v}_{n+2k_m}=(1,1)$.

\item[b)] Suppose that for some $i>n$ in the block $S_m$, for which $i-n$ is even, we have $s_i=0$.
Let $i$ be the first agent in the block with this property. We have $\mathbf{v}_{i}=(1,0)$ (as in the previous case), but $x_i=0$, so that $\mathbf{v}_{i+1}=(0,0)$. Then, all subsequent decisions in the block, as well as by the next SR transient agent are 0, and the first agent of the next block, $R_m$, observes $\mathbf{v}_{n+2k_m}=(0,0)$.

\end{itemize}

\end{enumerate}
 
 \old{
\begin{figure}
\centering
\includegraphics[scale=0.5]{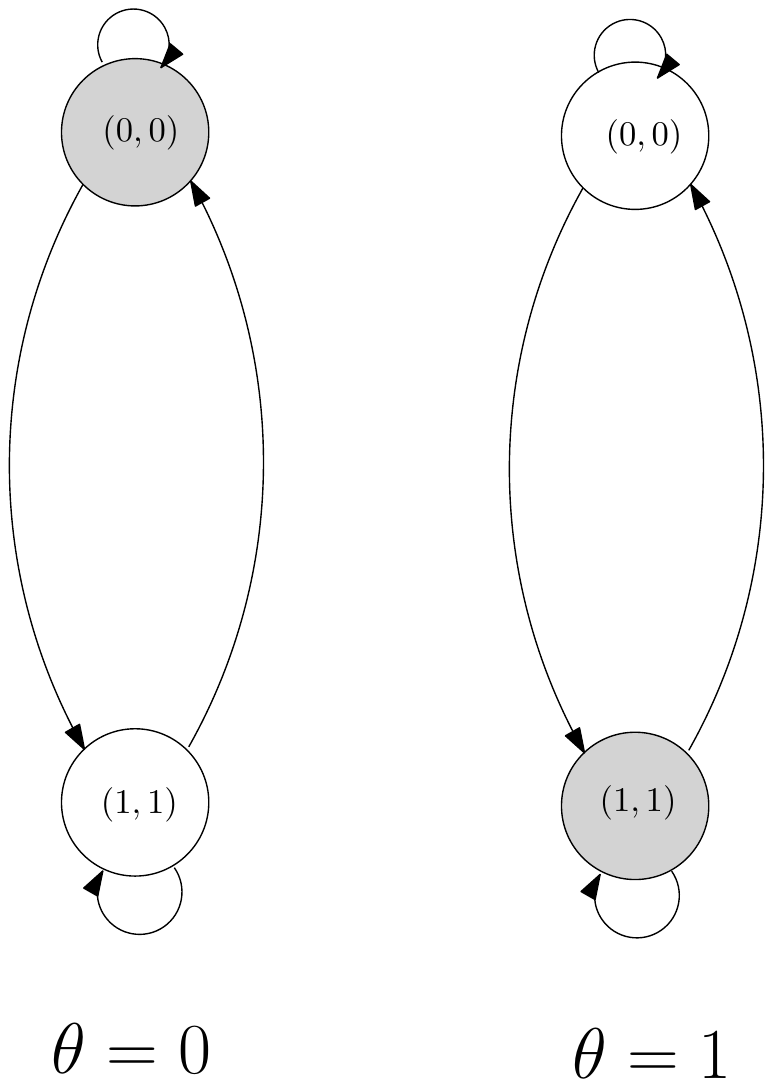}
\caption{The state observed by the first agent of subsequent blocks corresponds to an inhomogeneous Markov chain.}
\label{fig:markov_tn}
\end{figure}
}

To understand the overall effect of our construction, we consider a (non-homogeneous) Markov chain representation of the evolution of decisions.   We focus on the subsequence of agents consisting of the first agent of each S- and R-block. By the construction of the decision profile, the state $\xdn$, restricted to this subsequence, can only take values $(0,0)$ or $(1,1)$, and its evolution can be represented {by} a 2-state Markov chain. The transition probabilities between the states in this Markov chain is given by a product of terms, the number of which is related to the size of the S- and R-blocks.  For learning to occur, there has to be an infinite number of switches between the two states in the Markov chain (otherwise getting {trapped} in an incorrect decision would have positive probability). Moreover, the probability of these switches should go to zero (otherwise there would be a probability of switching to the incorrect decision that is bounded away from zero). We obtain these features by allowing switches from state $(0,0)$ to state $(1,1)$ during S-blocks and switches  from state $(1,1)$ to state $(0,0)$ during R-blocks. By suitably defining blocks of increasing size, we can ensure that the probabilities of such switches remain positive but decay at a desired rate. This will be accomplished by the parameter choices described next.

Let  $\log(\cdot)$ stand for the natural logarithm. 
For $m$ large enough so that $\log m$ is larger than both $1/\op$ and $1/\oq$, 
we let 
\begin{equation}\label{sblocks}
k_m=\left \lceil\log_{1/\op}\left(\log m\right ) \right \rceil,
\end{equation}
and
\begin{equation}\label{rblocks}
r_m=\left \lceil\log_{1/\oq}\left(\log m\right ) \right \rceil,
\end{equation}
both of which are positive numbers. Otherwise, for small $m$, we let $k_m=r_m=1$. These choices guarantee learning.

\begin{theorem}\label{thrm:algorproof}
Under the decision profile  and the parameter choices described in this section, 
\begin{equation*}
\lim_{n\rightarrow  \infty} \mathbb{P}(x_n=\theta)=1.
\end{equation*}
\end{theorem}

\subsection{Proof of Theorem \ref{thrm:algorproof}}

The proof relies on the following fact.
\begin {lemma} \label{thrm:converg-diverge}
Fix an integer $L\geq 2$. If $\alpha > 1$,  then the series 
\begin{equation*}
\sum_{m=L}^{\infty}\frac{1}{m \, \log^{\alpha}(m)},
\end{equation*}
converges; if $ \alpha \leq 1, $ then  the series diverges. \end{lemma}
\begin{proof}
See Theorem 3.29 of \cite{Rudin}.
\end{proof}

The next lemma characterizes the transition probabilities of the non-homogeneous Markov chain associated with the state of the first agent of each block.
For any $m\in\mathbb{N}$, let 
$w_{2m-1}$ be the decision of the last agent before block $S_m$, and let $w_{2m}$ be the decision of the last agent before block $R_m$. Note that for $m=1$, $w_{2m-1}=w_1$ is the decision $x_2=0$, since the first agent of block $S_1$ is agent $3$.
More generally, when $i$ is odd (respectively, even),  $w_i$ describes the state at the beginning of an S-block (respectively, R-block), and in particular, the decision of the transient agent preceding the block.

\begin{lemma}\label{lem:tranprob}
We have
$$
\mathbb{P}(w_{i+1}=1 \mid w_i=0){=} \begin{cases} p^{k_{m(i)}}\cdot \displaystyle{\frac{1}{{m(i)}}}, &\text{if} \ $i$  \text{ is odd,}\\ 
																									0, &\text{otherwise,}
																		\end{cases}
$$
 and
$$\mathbb{P}(w_{i+1}=0 \mid w_i=1){=} \begin{cases}q^{r_{m(i)}} \cdot \displaystyle{\frac{1}{{m(i)}}}, &\text{if} \ $i$ \ \text{ is even,}\\
																									0,&\text{otherwise,}
																		\end{cases}
$$
{where  } 
\[
m(i)=\begin{cases} {(i+1)}/{2},&\text{if} \ i \ \text{ is odd,} \\
   								{i}/{2},&\text{if} \  i \ \text{is even.}
   								\end{cases}
\]
(The above conditional probabilities are taken under either state of the world $\theta$, with the parameters $p$ and $q$ on the right-hand side being the corresponding probabilities that $s_n=1$ and $s_n=0$.)
\end{lemma}
\begin{proof} 
Note that $m(i)$ is defined so that $w_i$ is associated with the beginning of either block $S_{m(i)}$ or $R_{m(i)}$, depending on whether $i$ is odd or even, respectively. 

Suppose that $i$ is odd, so that we are dealing with the beginning of an S-block. If $w_i=1$, then,
as discussed in the previous subsection, we will have $w_i=1$, which proves that
$\mathbb{P}(w_{i+1}=0 \mid w_i=1){=} 0$.

Suppose now that $i$ is odd and $w_i=0$.
In this case, there exists only one particular sequence of events under which the  state will  change to $w_{i+1}=1$. Specifically, the searching phase should be initiated (which happens with probability $1/m(i)$), and the private signals of {about} half of the agents in the block $S_{m(i)}$ ($k_{m(i)}$ of them)  should be equal to $1$. The probability of this sequence of events is precisely the one given in the statement of the lemma.
 
The transition probabilities for the case where $i$ is even are obtained by a symmetrical argument.
\end{proof}

The reason behind our definition of $k_m$ and $r_m$ is that we wanted to enforce 
Eqs.~\eqref{cond1}-\eqref{cond2} in the lemma
 that follows. 

\begin{lemma}\label{lem:algerbr}
We have
\begin{equation}
\sum_{m=1}^{\infty} p_1^{k_m} \frac{1}{m}=\infty, \qquad
\sum_{m=1}^{\infty} q_1^{r_m} \frac{1}{m}<\infty, \label{cond1}
\end{equation}
and
\begin{equation}
\sum_{m=1}^{\infty} p_0^{k_m} \frac{1}{m}<\infty,\qquad
\sum_{m=1}^{\infty} q_0^{r_m} \frac{1}{m}=\infty.\label{cond2}
\end{equation}

\end{lemma}

\begin{proof}
For $m$ large enough, the definition of $k_m$ implies that 
\begin{equation*}
\log_{\op}\left(\frac{1}{\log m}\right)\leq k_m < \log_{\op}\left(\frac{1}{\log m}\right) +1,
\end{equation*}
or equivalently, 
\begin{equation*}
p \cdot p^{\log_{\op}\left( \frac{1}{\log m}\right) }< p^{k_m} \leq  p^{\log_{\op}\left( \frac{1}{\log m}\right) },
\end{equation*}
where $p$ stands for either $p_0$ or $p_1$.
(Note that the direction of the inequalities was reseversed because the base $\op$ of the logarithms is less than $1$.)
Dividing by $m$, using the identity $p=\op^{\log_{\op}(p)}$, after some elementary manipulations, we obtain
\[
p\frac{1}{m \log^{\alpha}{m}}< p^{k_m} \frac{1}{m} 
\leq  \frac{1}{m \log^{\alpha}{m}},   
\]
where $\alpha= \log_{\op}(p)$.
By a similar argument, 
\[
q \frac{1}{m\log^{\beta}{m}}< q^{k_m} \frac{1}{m} 
\leq  \frac{1}{m\log^{\beta}{m}},
\]
where $\beta= \log_{\oq}(q)$.

Suppose  that  $p=p_1$, so that $p>\op$ and $q<\oq$. 
Note that $\alpha$ is a decreasing function of $p$, because the base of the logarithm satisfies $\op<1$. Since $\log_{\op}(\op)=1$, it follows that {$\alpha=\log_{\op}(p)<1$, and by a parallel argument, $\beta>1$.} 
Lemma \ref{thrm:converg-diverge} then implies that  conditions  (\ref{cond1}) hold.
Similarly, if $p=p_0$, so that $p<\op$ and $q>\oq$,  then {$\alpha >1$ and $\beta <1$,} and
conditions  (\ref{cond2}) follow again from
Lemma \ref{thrm:converg-diverge}.
\end{proof}

We are now ready to complete the proof, using a standard Borel-Cantelli argument.

 \begin{proof} [Proof of Theorem \ref{thrm:algorproof}]
Suppose that $\theta=1$. Then, by Lemmata \ref{lem:tranprob} and  \ref{lem:algerbr}, we have  that 
\[
\sum_{i=1}^\infty \mathbb{P}^{{1}}(w_i=1 \mid w_i=0) =\infty,
\]
while
\[
\sum_{i=1}^\infty \mathbb{P}^{{1}}(w_{i+1}=0\mid w_i=1) <\infty.
\]
Therefore,  transitions from the state $0$ of the Markov chain $\{w_i\}$ to  state $1$ are guaranteed to happen,  while transitions from state {$1$} to state {$0$} will happen only finitely many times. It follows that $w_i$ converges to $1$, almost surely, when $\theta =1$. 
By a symmetrical argument, $w_i$ converges to $0$, almost surely, when $\theta=0$.

Having proved (almost sure) convergence of the sequence $\{w_i\}_{i \in\mathbb{N}}$, it remains to prove convergence (in probability) of the sequence $\{x_n\}_{n \in \mathbb{N}}$ (of which
$\{w_i\}_{i \in\mathbb{N}}$ is a subsequence). This is straightforward, and we only outline the argument. If $w_i$ is the decision $x_n$ at the beginning of a segment, then $x_n=w_i$ for all $n$ during that segment, unless a searching phase is initiated. A searching phase gets initiated with probability at most $1/m$ at the beginning of the S-block and with probability at most $1/m$ at the beginning of the R-block. Since these probabilities go to zero as $m\to\infty$, it is not hard to show that $x_n$ converges in probability to the same limit as $w_i$.

\end{proof}
\old{
%%% There is a sharp contrast between the learning results for $K=1$ and $K>1$. The reason is best understood by considering the relation between transition probabilities of the Markov chain representations of decision processes. For $K=1$, under the BLR assumption, the transition probabilities between the states $0$ and $1$ under the two states of the world are coupled (see Lemma \ref{cor:blrlemma}). This is central to establishing that learning is not possible under any decision rule. For $K=2$, the Markov chain that describes the evolution of the decisions of agents has four states and the transition probabilities between the states $(0,0)$ and $(1,1)$ can be designed to be a product of intermediate terms whose number is a decreasing function of time. Using this construction we can ensure that the ratio of the transition probabilities under the two states of the world asymptotically reach zero or infinity, similar to what one would obtain for signal structures with Unbounded Likelihood Ratios.
}

The existence of a decision profile that guarantees learning in probability naturally leads to the question of providing incentives to agents to behave accordingly. It is known  \cite{Cov69,AtPa89,Acem10} that for Bayesian agents who minimize the probability of an erroneous decision, learning in probability does not occur, which brings up the question of designing a game whose equilibria have desirable learning properties. A natural choice for such a game is explored in the next section, although our results will turn out to be negative. 

\section{Forward looking agents}\label{se:fwd}
In this section, we assign to each agent a payoff function  that depends on its own decision as well as on future decisions. We consider the resulting game between the agents and study the learning properties of the equilibria of this game. In particular, we show that learning fails to obtain at any of these equilibria.

\subsection{Preliminaries and notation}

In order to conform to game-theoretic terminology, we will now talk about strategies $\sigma_n$ (instead of decision rules $d_n$). 
A (pure) \textbf{strategy} for agent $n$ is a mapping $\sigma_n :  \{0,1\}^K \times S\rightarrow \{0,1\}$ from the agent's information set (the vector $\xdn=(x_{n-1},\ldots,x_{n-K})$ of decisions of its $K$ immediate predecessors and its private signal $s_n$) to a binary decision, so that 
$x_n=\sigma_n(\xdn,s_n)$. 
A \textbf{strategy profile} is a sequence of strategies, $\sigma=\{\sigma_n\}_{n\in \mathbb{N}}$.  We  use the standard notation $\sigma_{-n}=\{\sigma_1,\ldots, \sigma_{n-1},\sigma_{n+1},\ldots\}$ to denote the collection of strategies of all agents other than $n$, so that $\sigma=\{\sigma_{n},\sigma_{-n}\}$. 
Given a strategy profile $\sigma$, the resulting sequence of decisions $\{x_n\}_{n \in \mathbb{N}}$ is a well defined stochastic process.
% and we denote the measure  generated by this stochastic process by $\mathbb{P}_\sigma$.

 The payoff function of agent $n$ is
\begin{equation} 
\sum_{k=n}^{\infty}\delta^{k-n} \mathbf{1}_{x_k=\theta}, \label{payoff}
\end{equation}
where $\delta \in (0,1)$ is a discount factor, and $\mathbf{1}_A$ denotes the indicator random variable of an event $A$. 
Consider some agent $n$ and suppose that the strategy profile $\sigma_{-n}$ of the remaining agents has been fixed. Suppose that agent $n$ observes a particular vector $\u$ of predecessor decisions (a realization of $\xdn$) and a realized value $s$ of the private signal $s_n$. Note that $(\xdn,s_n)$ has a well defined distribution once $\sigma_{-n}$ has been fixed, and can be used by agent $n$ to construct a conditional distribution (a posterior belief) on $\theta$. Agent $n$ now considers the two alternative decisions, $0$ or $1$. For any particular decision that agent $n$ can make, the decisions of subsequent agents $k$ will be fully determined by the recursion $x_k=\sigma_n(\mathbf{v}_k,s_k)$, and will also be well defined random variables. This means that the conditional expectation of agent $n$'s payoff, if agent $n$ makes a specific decision $y\in\{0,1\}$,
\begin{align*}
U_n&(y; \mathbf{u}, s) \nonumber\\
&=\mathbb{E}\left[\mathbf{1}_{\theta=y}+\sum_{k=n+1}^{\infty}\delta^{k-n} \mathbf{1}_{x_k=\theta} \ \Big| \  \xdn=\mathbf{u},  s_n=s\right],
\end{align*}
is unambiguously defined, modulo the usual technical caveats associated with conditioning on zero probability events; 
in particular, the conditional expectation is uniquely defined for ``almost all'' $(\mathbf{u}, s)$, that is, modulo on a set of $(\xdn,s_n)$ values that have zero probability measure under $\sigma_{-n}$.
We can now define our notion of equilibrium, which requires that given the decision profile of the other agents, each agent maximizes its conditional expected payoff $U_n(y; \mathbf{u}, s)$  over $y\in\{0,1\}$, for almost all $(\mathbf{u}, s)$.

\begin{definition}\label{d:equil}
A strategy profile $\sigma$ is an \textbf{equilibrium}  if for each $n \in \mathbb{N}$, for each vector of observed actions $\mathbf{u} \in \{0,1\}^K$ that can be realized under $\sigma$ with positive probability (i.e., $\Pr(\mathbf{v}_n=\mathbf{u})>0$), and 
for almost all $s\in S$, $\sigma_n$ maximizes the expected payoff of agent $n$ given the strategies of the other agents, $\sigma_{-n}$, i.e.,\[
\sigma_n(\mathbf{u}, s)\in \argmax_{y \in \{0,1\}} U_n(y, \mathbf{u}, s).
\]
\end{definition}

Our main result follows.

\begin{theorem} \label{thrm:frwrdlooking}
For any discount factor $\delta \in [0,1)$ and for any equilibrium strategy profile, learning fails to hold.
\end{theorem}

We note that the set of equilibria, as per Definition \ref{d:equil}, contains the Perfect Bayesian Equilibria,   as defined in \cite{FuTi}. Therefore, Theorem \ref{thrm:frwrdlooking} implies that there is no learning at any  Perfect Bayesian Equilibrium.

From now on, we assume that we fixed a specific strategy profile $\sigma$.
Our analysis centers around the case where an agent observes a sequence of ones from its immediate predecessors, that is, $\vn=\e$, where $\e=(1, 1, \ldots, 1)$.  The posterior probability that the state of the world is equal to $1$, based on having observed a sequence of ones  is defined by 
\[
\pi_n=\Pr(\theta=1 \mid \vn=\e).
\]
Here, and in the sequel, we use $\Pr$ to indicate probabilities of various random variables under the distribution induced by $\sigma$, and similarly for the conditional measures $\Pr^j$ given that the state of the world is $j\in\{0,1\}$.
For any private signal value $s \in S$, we  also define 
\[
f_n(s)=\Pr(\theta=1 \mid \vn=\e, s_n=s).
\]
Note that these conditional probabilities are well defined as long as $\Pr(\vn=\e)>0$ and for almost all $s$. We also let
\[
f_n=\textrm{essinf}_{s \in S} f_n(s).
\]
Finally, for every agent $n$, we define the \emph{switching probability} under the state of the world $\theta=1$, by 
\[
\gamma_n=\Pr^1(\sigma_n(\e,s_n)=0).
\]

{We will prove our result by contradiction, and so we assume} that $\sigma$ is an equilibrium that achieves learning in probability. In that case, under state of the world $\theta=1$,  all agents will eventually be choosing $1$ with high probability. Therefore, when $\theta=1$, blocks of size $K$ with all agents choosing $1$ (i.e., with $\mathbf{v}_n=\mathbf{e}$) will also occur with high probability. 
The Bayes rule will then imply that  
 the posterior probability that $\theta=1$, given that $\mathbf{v}_n=\mathbf{e}$,  will eventually be arbitrarily close to one.  The above are formalized in the next Lemma. 

\begin{lemma} \label{lem:frwrdproperties}
{Suppose that the strategy profile $\sigma$ leads to learning in probability. Then,}
\begin{enumerate}[(i)]
\item $\lim_{n \rightarrow \infty}\Pr^0(\mathbf{v}_n=  \e)=0$ and\/ $\lim_{n \rightarrow \infty}\Pr^1(\mathbf{v}_n=  \e)=1$.

\item $\lim_{n \rightarrow \infty} \pi_n =1$,

\item $\lim_{n \rightarrow \infty}f_n(s)=1,$ { uniformly over all } $s \in S$, except possibly on a zero  measure subset of $S$.

\item $\lim_{n \rightarrow \infty} \gamma_n=0$.

\end{enumerate}
\end{lemma}

\begin{proof}
\begin{enumerate}[(i)]
\item Fix some $\epsilon >0$. By the learning in probability assumption, 
$$\lim_{n \rightarrow \infty}\Pr^0(\mathbf{v}_n=  \e)
\leq \lim_{n \rightarrow \infty}\Pr^0(x_n=  1)=0.$$
Furthermore, 
there exists ${N} \in \mathbb{N}$ such that for all $n>{N}$,
\[
\mathbb{P}^1(x_n=0 )<\frac{\epsilon}{K}.
\]
Using the union bound, we obtain 
\[
\mathbb{P}^1(\mathbf{v}_n= \mathbf{e}) \geq 1 -\sum_{k=n-K}^{n-1} \Pr^1(x_k=0) > 1-\epsilon,
\]
for all $n>N+K$. Thus,  $\lim_{n \rightarrow \infty}\Pr^1(\mathbf{v}_n=  \e)>1-\epsilon$. Since $\epsilon$  is arbitrary, the result for $\Pr^1(\mathbf{v}_n=  \e)$ follows.

\item  Using the Bayes rule and the fact that the two values of $\theta$ are a priori equally likely, we have 
\[
\pi_n=\frac{\mathbb{P}^1(\vn=\e)}{\mathbb{P}^0(\vn=\e)+\mathbb{P}^1(\vn=\e )}.
\]
The result follows from part (i).

\item Since the two states of the world are a priori equally likely, the ratio ${f_n(s)}/({1-f_n(s)})$ of posterior probabilities, is equal to the likelihood ratio associated with the information $\mathbf{v_n}=\e$ and $s_n=s$, i.e, 
\[
\frac{f_n(s)}{1-f_n(s)}=\frac{\Pr^1(\vn=\e)}{\Pr^0(\vn=\e)} \cdot \frac{d\mathbb{F}_1}{d\mathbb{F}_0}(s),
\]
almost everywhere, where we have used the independence of $\mathbf{v}_n$ and $s_n$ under either state of the world. 
Using the BLR assumption, 
\[
\frac{f_n(s)}{1-f_n(s)}\geq\frac{1}{{M}}\cdot \frac{\Pr^1(\vn=\e)}{\Pr^0(\vn=\e)}.
\]
almost everywhere. Hence, 
using the result in part (i), 
\[
\lim_{n \rightarrow \infty}\frac{f_n(s)}{1-f_n(s)}=\infty, 
\]
uniformly over  all $s \in S$,  except possibly over a countable union of zero measure sets (one zero measure set for each $n$). It follows that
 $\lim_{n \rightarrow \infty}f_n(s)=1$, uniformly over $s \in S$, except possibly on a zero measure set.

\item We note that
\[
\Pr^1(x_n=0, \vn=\e ) =\Pr^1(\vn=\e) \cdot 
\gamma_n.
\]
Since $\Pr^1(x_n=0, \vn=\e )\leq \Pr^1(x_n=0)$, we have 
$\lim_{n \rightarrow \infty} \Pr^1(x_n=0, \vn=\e ) =0$. Furthermore, from part (i), 
$\lim_{n \rightarrow \infty}\Pr^1(\mathbf{v}_n=  \e)=1$. It follows that 
$\lim_{n\to\infty} \gamma_n=0$.
\end{enumerate} \end{proof}

We now proceed to the main part of the proof. We will argue that under the learning assumption, and in the limit of large $n$, it is more profitable for agent $n$ to choose $1$ when observing a sequence of ones from its immediate predecessors, rather than choose $0$, irrespective of its private signal. This implies that after some finite time $N$, the agents will be copying their predecessors' action, which is inconsistent with learning.  

\begin{proof}[Proof of Theorem \ref{thrm:frwrdlooking}]
Fix some $\epsilon \in (0,1-\delta)$. 
%By Lemma \ref{lem:frwrdproperties}(iv), there exists some $\tilde{N} \in \mathbb{N}$ such that for all $n>\tilde{N}$, $\gamma_n<\epsilon/2$.
%For $n>\tilde{N}$, 
We define 
\[
t_n={\sup}\Big\{t: \sum_{k=n}^{n+t}\gamma_k\leq \epsilon\Big\}.
\]
(Note that $t_n$ can be, in principle, infinite.) 
Since $\gamma_k$ converges to zero (Lemma \ref{lem:frwrdproperties}(iv)), it follows that 
$\lim_{n \rightarrow \infty}t_n=\infty$.  

Consider an agent $n$ who observes $\vn=\e$ and  $s_n=s$, and who makes a decision $x_n=1$. 
(To simplify the presentation, we assume that $s$ does not belong to any of the exceptional, zero measure sets involved in earlier statements.)
The (conditional) probability that agents   $n+1,\ldots, n+t_n$ all decide $1$  is 
\begin{align*}
\Pr&\left(\bigcap_{k={n+1}}^{n+t_n} \left \{\sigma_k(s_k, \e)=1\right\} \right)  = \prod_{k=n+1}^{n+t_n}\left(1-\gamma_k\right)\\
&\geq 1- \sum_{k=n+1}^{n+t_n}\gamma_k\geq 1-\epsilon.
\end{align*}
With agent $n$ choosing the decision $x_n=1$, 
its payoff can be lower bounded by considering only the payoff obtained when $\theta=1$ (which, given the information available to agent $n$, happens with probability $f_n(s)$) and all agents up to $n+t_n$ make the same decision (no switching):
\[
U_n(1;\e,s) \geq f_n(s)\left( \sum_{k=n}^{n+t_n} \delta^{k-n}  \right) (1-\epsilon).
\]
Since $f_n(s)\leq 1$ for all $s \in S$, and 
\[
\sum_{k=n}^{n+t_n} \delta^{k-n}  \leq \frac{1}{1-\delta},
\]
we obtain 
\[
U_n(1;\e,s)\geq f_n(s)\left( \sum_{k=n}^{n+t_n} \delta^{k-n}  \right) -\frac{\epsilon}{1-\delta}.
\]
Combining with part (iii) of Lemma \ref{lem:frwrdproperties} and the fact that $\lim_{n \rightarrow \infty}t_n=\infty$, we obtain
\begin{equation}
\liminf_{n\rightarrow \infty}U_n(1; \e, s) \geq \frac{1}{1-\delta}-\frac{\epsilon}{1-\delta}. \label{eq:limitinf}
\end{equation}

On the other hand, the payoff from deciding $x_n=0$ can be bounded from above  as follows:
\begin{align*}
U_n&(0;\e,s)\\
&=\mathbb{E}\left[\mathbf{1}_{\theta=0}+\sum_{k=n+1}^{\infty}\delta^{k-n} \mathbf{1}_{x_k=\theta}\ \Big| \ \xdn=\e, s_n=s\right]\\
   									 &\leq \Pr(\theta=0\mid \xdn=e,s_n=s)+\frac{\delta}{1-\delta} \\
   									 &=1-f_n(s)+\frac{\delta}{1-\delta}.
\end{align*}
Therefore, using part (iii) of Lemma \ref{lem:frwrdproperties}, 
\begin{equation}
\limsup_{n \rightarrow \infty} U_n(0;\e,s) \leq \frac{\delta}{1-\delta}. \label{eq:limitsup}
\end{equation}
Our choice of $\delta$ implies that 
\[
\frac{1}{1-\delta}-\frac{\epsilon}{1-\delta}>\frac{\delta}{1-\delta}.
\]
Then, (\ref{eq:limitinf}) and (\ref{eq:limitsup}) imply that there exists $N \in \mathbb{N}$ such that for all $n >N$,
\[
U_n(1; \e, s)>U_n(0;\e,s).
\]
almost everywhere in $S$.
Hence, by the equilibrium property of the strategy profile $\sigma_n(e,s)=1$ for all $n >N$ and for all $s \in S$, except possibly on a zero measure set. 

Suppose that the state of the world is $\theta=1$. Then, by part (i) of Lemma \ref{lem:frwrdproperties}, 
$\xdn$ converges to $\e$, in probability, and therefore it converges to $\e$ almost surely along a subsequence. In particular, the event 
$\{\xdn=\e\}$ happens infinitely often, almost surely. If that event happens and $n>N$,  
then every subsequent $x_k$ will be equal to 1. Thus, $x_n$ converges almost surely to $1$. By a symmetrical argument,  if $\theta=0$, then $x_n$ converges almost surely to $0$. Therefore, 
$x_n$ converges almost surely to $\theta$. This is 
impossible, by Theorem \ref{thm:almsure}. We have reached a contradiction, thus establishing that learning in probability fails under the equilibrium strategy profile $\sigma$. \end{proof}

\section{Conclusions}\label{se:conc}
We {have obtained sharp} results on the fundamental limitations of learning by a sequence of agents who only get to observe the decisions of a fixed number $K$ of immediate predecessors, under the assumption of Bounded Likelihood Ratios. Specifically, we have shown that almost sure  learning is impossible whereas learning in probability is possible if and only if $K>1$.  We then 
studied the learning properties of the equilibria of a game where agents are forward looking, with a discount factor $\delta$ applied to to future decisions. As $\delta$  ranges in $[0,1)$ the resulting strategy profiles vary from the myopic ($\delta=0$) towards the case of fully aligned objectives ($\delta\rightarrow 1$). Interestingly, under a full alignment of objectives and a central designer, learning is possible when $K\geq 2$, yet learning fails to obtain at any equilibrium of the associated game, and for any 
$\delta\in [0,1)$.

The scheme in Section \ref{sec:learnalgo} is only of theoretical interest, because the rate at which the probability of error decays to zero is extremely slow. This is quite unavoidable, even for the much more favorable case of unbounded likelihood ratios \cite{TTW08}, and we do not consider the problem of improving the convergence rate a promising one.

The existence of a decision profile that guarantees learning in probability (when $K\geq 2$) naturally leads to the question of whether it is possible to provide incentives to the agents to behave accordingly. It is known  \cite{Cov69,AtPa89,Acem10} that for myopic Bayesian agents, learning in probability does not occur, which raises the question of designing a game whose equilibria have desirable learning properties.  Another interesting direction is the characterization of the structural properties of decision profiles that allow or prevent learning whenever the latter is achievable. 

{Finally, one may consider extensions to the case of $m>2$ hypotheses and $m$-valued decisions by the agents. Our negative results are expected to hold, and the construction of a decision profile that learns when $K\geq m$, is also expected to go through, paralleling a similar extension in \cite{Kop75}.}

 \bibliography{main}
 \bibliographystyle{plain}

\end{document}